\documentclass[journal,twocolumn]{IEEEtran}
\usepackage[utf8]{inputenc}
\newcounter{MYtempeqncnt}
\usepackage{psfrag}
\usepackage{pstool}
\usepackage{hyperref}

\usepackage{amsmath}
\usepackage{bbm}
\usepackage{amssymb}
\usepackage{graphicx}
\usepackage{placeins}
\usepackage{color,soul}

\usepackage{subfigure}
\usepackage{amsthm}
\usepackage{amsmath}
\usepackage{mathtools}
\usepackage{amsfonts}
\def\Pr{\mathbb{P}}
\usepackage{verbatim}
 \usepackage{lastpage}
\def\Pr{\mathbb{P}}
\def\peq#1{\stackrel{\text{\scriptsize(#1)}}{=}}

\def\ind{\mathbbm{1}}
\def\Ex{\mathbb{E}}
\def\Var{\mathbb{V}}
\def\dd{\mathrm{d}}

\theoremstyle{definition}
\newtheorem{theorem}{Theorem}
\newtheorem{corollary}{Corollary}
\newtheorem{lemma}{Lemma}
\newtheorem{remark}{Remark}

\usepackage[style=ieee]{biblatex}
    \bibliography{ref.bib}

 \title{{Spatial Disease Propagation With Hubs}}
\author{\IEEEauthorblockN{Ke Feng}
and
\IEEEauthorblockN{Martin Haenggi,}~\IEEEmembership{Fellow,~IEEE}
\thanks{Ke Feng is with Laboratoire ETIS-UMR 8051, CY Cergy Paris Université, ENSEA, CNRS, Cergy-Pontoise, France. The work of K. Feng was supported by the ERC NEMO grant, under the European Union's Horizon 2020 research
and innovation programme, grant agreement number 788851 to INRIA. E-mail: \href{mailto:ke.feng@ensea.fr}{ke.feng@ensea.fr}. Martin Haenggi is with the Dept. of Electrical Engineering, University of Notre Dame, Indiana, USA. E-mail: \href{mailto:mhaenggi@nd.edu}{mhaenggi@nd.edu}.}}
\begin{document}
\maketitle
\begin{abstract}
    Physical contact or proximity is often a necessary condition for the spread of infectious diseases.
Common destinations, typically referred to as hubs or points of interest, are arguably the most effective spots for the type of disease spread via airborne transmission. In this work, we model the locations of individuals (agents) and common destinations (hubs) by random spatial point processes in $\mathbb{R}^d$ and
focus on disease propagation through agents visiting common hubs. The probability of an agent visiting a hub depends on their distance through a connection function $f$. The system is represented by a random bipartite geometric (RBG) graph. We study the degrees and  percolation of the RBG graph for general connection functions. We show that the critical density of hubs for percolation is dictated by the support of the connection function $f$, which reveals the critical role of long-distance travel (or its restrictions) in disease spreading.
\end{abstract}
\begin{IEEEkeywords}
Stochastic geometry, epidemic networks, disease propagation, random geometric bipartite graph, degrees, percolation
\end{IEEEkeywords}

\section{Introduction}
\subsection{Motivation}
The spread of many airborne diseases relies on physical contact or proximity between susceptible individuals. Individuals visiting common destinations are especially likely to further diseases' propagation. For humans, such destinations can be supermarkets, social gatherings, airports, etc. Understanding the joint effect of individual density, density of the common destinations, and their connectivity pattern is a first step to the prevention of disease spreading.
To this end, we propose and study a
 mathematical model for spatial disease propagation with hubs. The model is based on \textit{random bipartite geometric} (RBG) graphs (defined in Sec II.A) and generalizes the AB continuum model \cite{iyer2012percolation,penrose_2014}.
In an RBG graph, there are two sets of entities,  referred to as \textit{agents} and \textit{hubs} respectively in this work, each being distributed in the Euclidean space according to a random spatial point process, hence the random and geometric aspect of the graph. Edges may only exist between the two sets of entities and not within, hence the bipartite aspect of it. {Lastly, the connectivity pattern between agents and hubs usually depends on their geometric distance. Dispersion of the connectivity pattern can be used to model a change in the behavior of the agents while keeping other parameters fixed, or to compare different scenarios where agents tend to travel farther or stay closer to their home location. }

\subsection{Related Work}
One of the earliest and most popular type of epidemic models are homogeneous mixture models, where contacts between all pairs of nodes are equally likely \cite{newman2018networks,pastor2015epidemic}. In these models, a key indicator is the basic reproduction number \cite{Delamater2019ComplexityOT}, defined as the expected number of secondary infections caused by a single infection. Such models do not reflect the heterogeneous nature of human interactions, which are often also location dependent.
A natural approach is to introduce spatial modeling \cite{Li2011TheFO,pastor2015epidemic,BARTHELEMY20111,katori2021continuum}. In spatial modeling, nodes in proximity have interactions and make it possible for infections to occur. The question is if or when the disease, through local connectivity, eventually spreads to a large component of the nodes when originated from a single node. Such questions are about the critical thresholds such that percolation occurs \cite{frisch1963percolation}.
The most known setups perhaps are square lattices with Bernoulli percolation in the discrete case and Boolean models based on Poisson point processes (PPPs) in the continuum. 
The latter is generalized by Poisson random connection models, which assume that two nodes $x,x'$ in a PPP connect with probability $f(\|x-x'\|)\in[0,1]$, where $f(\cdot)$ is the connection function \cite[Chapter 6]{MeesterRonald1996CP}. {For Poisson random connection models, it is shown that dispersive functions, which correspond to long-range but  unreliable connections, make 
percolation easier }\cite{Franceschetti05continuumpercolation,penrose1993spread}.
More recently, \cite{baccelli2020computational1} considers the combination of disease dynamics and random geometry through studying the steady state of an susceptible-infected-susceptible (SIS) model on spatial point processes. It considers transmissions between nodes within a given radius as well as the effect of node mobility.

 The role of hubs (common destinations) in disease spreading is not reflected in the models above. In this respect, an epidemic model based on an abstract random bipartite graph is studied in \cite{britton2008epidemics}. It constructs the random bipartite graph by drawing an edge between any pair of vertices of distinct types with a fixed probability, then extracts its unipartite graph, and derives the probability of explosion under tunable clustering. The model is homogeneous and does not consider geometric dependence. The model most relevant to our setting is the random bipartite geometric model in \cite{iyer2012percolation,penrose_2014}. It is also called AB continuum percolation model due to its discrete counterpart, AB percolation models on, e.g., $\mathbb{Z}^2$ \cite[Chapter 12]{grimmett1999percolation}. It is first proposed and studied by \cite{iyer2012percolation}, where the authors consider two independent PPPs $\Phi,~\Psi$, defined on $\mathbb{R}^2$, and edges are drawn between $x\in\Phi$ and $y\in \Psi$ if $x,y$ are within a prescribed distance. The motivation of the model comes from communication networks with two types of transmitters. The authors proved the percolation thresholds (defined in Sec IV) and their bounds. Further results for this model are proved in \cite{penrose_2014}. The RBG model studied here generalizes the AB continuum model in the same way that Poisson random connection models generalize Poisson Boolean models.

\subsection{Contribution}
The contributions of this work are summarized as follows.
\begin{itemize}
    \item We propose a mathematical model based on the random bipartite geometric (RBG) graph for disease propagation through agents visiting hubs.
    \item We analyze the statistics of the degree of the RBG graphs for general and Poisson point processes in $\mathbb{R}^d$. We show the impact of the dispersion of the connection function on these statistics.
    \item We show that the existence of a critical hub density for percolation depends only on the boundedness of the support of the connection function, which shows the necessity to curb long-distance travel to prevent disease spreading. Further, we give bounds on the critical densities of hubs and agents for the RBG graph with general connection functions. Lastly, we show that increasing the dispersion of the connection function lowers the critical threshold.
\end{itemize}

\section{System Model}
\subsection{The RBG Graph}
We model agents and hubs locations by $\Phi$ and $\Psi$, two jointly stationary and ergodic spatial point processes on $\mathbb{R}^d$, defined on a probability space $(\Omega,\mathcal{A},\mathbb{P})$.
Let  $\lambda$ and $\mu$ denote their intensities, respectively. We focus on the case where $\Phi$ and $\Psi$ are independent. Denote by $I(x,y)\in\{0,1\},~x\in\Phi,y\in\Psi$ the  indicator function of the undirected edge $\{x,y\}$, which is a Bernoulli random variable. $I(x,y)=1$ when there is an edge between $x,y$ and is zero otherwise. We assume that the distribution of $I(x,y)$ is independent of other edge indicators given $\Phi,~\Psi$. Further,
\begin{equation}
    \Pr(I(x,y)=1) =  f(\|x-y\|),\quad x\in\Phi,~y\in\Psi,
\end{equation} where $\|\cdot\|$ is the Euclidean norm and $f\colon \mathbb{R}^+\to[0,1]$ is referred to as {the connection function}. Let the edges of the graph be denoted by $\mathcal{E}$. We have
\[\mathcal{E} = \{\{x,y\}: I(x,y)=1,x\in\Phi,y\in\Psi\}.\] For practical relevance, we assume that $f$ is non-increasing and that $\int_{0}^{\infty} f(r) r^{d-1}\dd r <\infty$; the latter implies that the probability of an agent having unboundedly many edges is zero. 

The above defines an RBG graph, which we denote by $\mathcal{G}(\Phi,\Psi, f)$. 
 With a slight abuse of terminology, we say that two agents are \textit{connected} whenever there is a direct (two-edge) path  in $\mathcal{G}(\Phi,\Psi, f)$. From this definition, one can extract its corresponding unipartite geometric graph consisting of only the agents, where an edge exists between two connected agents. 

By contrast, a random (unipartite) geometric graph is generated on a spatial point process $\Phi$ on $\mathbb{R}^d$ with a connection function $f\colon \mathbb{R}^+\to[0,1]$, where $\mathcal{E} = \{\{x,y\}: I(x,y)=1,x\neq y,~x,y\in\Phi\}$ and $\Pr(I(x,y)=1) =  f(\|x-y\|),~x\neq y,~x,y\in\Phi$. 
We denote this unipartite graph by $\mathcal{G}(\Phi,f)$. 
The percolation of $\mathcal{G}(\Phi,f)$ for  a PPP $\Phi$ is studied in \cite[Chapters 3\&6]{MeesterRonald1996CP}, both for $f(r) = \ind (r\leq a)$, known as the Poisson Boolean model, and a generalized connection function, known as the Poisson connection model.

 Fig. \ref{fig:illu} shows snapshots of the RBG graph under two connection functions, one being the Boolean type, and the other being exponential. Both agent and hub locations are realizations of PPPs. 

  To model hubs of different popularity and density, the current model can be generalized to a heterogeneous spatial model similar to that of a multi-tier cellular network, which we will not discuss here.
  \begin{figure*}[t]
       \centering
       \subfigure[$f_1(r) = \ind(r\leq 0.1262)$]{         \includegraphics[width=.42\textwidth]{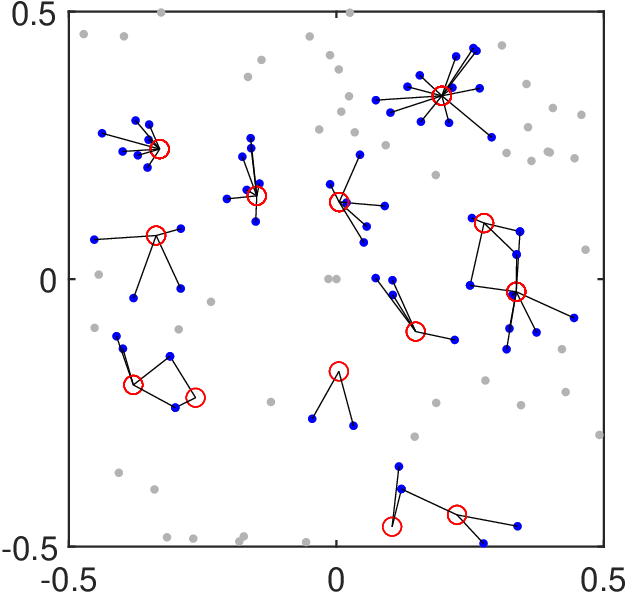}
         \label{fig: illu-f}}
     \hfill
     \subfigure[$f_2(r) = 1/2\exp(r/ 0.1262).$]{         \includegraphics[width=.42\textwidth]{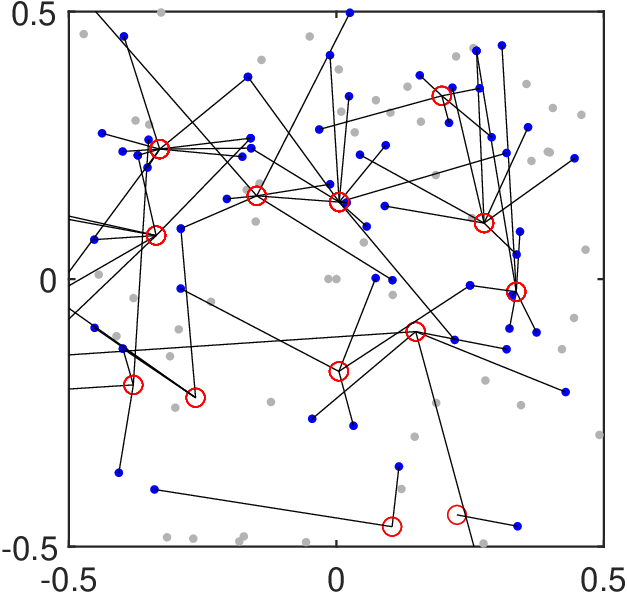}
         \label{fig: illu-g}}
       \caption{Illustrating one realization of the RBG graph on $\mathbb{R}^2$, where only its segment in a square window $[-0.5,0.5]^2$ is presented. Hubs are marked with circles, agents with dots, and edges with line segments. Agents without edges are in light gray. In both figures, the density of agents is 100 and the density of hubs is 10; the average number of edges of the typical hub is 5.  }
       \label{fig:illu}
   \end{figure*} 
\subsection{Dispersed RBG Graph}
 For a connection function $f$ and $p\in(0,1]$, we define its dispersed version
\begin{equation}
f_{p}(r) \triangleq p f(\sqrt[d]{p}r).
\end{equation}
{The dispersed version makes the connection function more spread-out and maintains the monotonicity. For example, when the support of $f$ is finite, dispersed version $f_p$ has expanded support. Edges with distance in the original support become less
reliable due to the monotonicity of $f$.}
It is easy to see that
$\int_0^{\infty} f(r)r^{d-1} \dd r = \int_{0}^{\infty} f_p(r)r^{d-1} \dd r$, i.e., the mean degrees of agents and hubs do not change. Specially, for $d=2$, we have $f_{p}(r) =p f(\sqrt{p}r)$, as defined in \cite{Franceschetti05continuumpercolation}. The RBG graph defined by the dispersed version of  connection function $f$ is $\mathcal{G}(\Phi,\Psi, f_p)$.



%



    \section{Degrees of the RBG Graph}
In this section, we consider some basic properties of the RBG graph $\mathcal{G}(\Phi,\Psi, f)$. Let $\Phi^o\triangleq \Phi\mid o\in\Phi.$ We focus on the typical agent at the origin and the Palm distribution of $\Phi$
\begin{equation}
 \mathrm{P}_{}^{o}(\cdot) \triangleq \Pr(\Phi\in\cdot \mid o\in\Phi),\nonumber
\end{equation}
and the reduced Palm distribution
\begin{equation}
 \mathrm{P}_{}^{!o}(\cdot) \triangleq \Pr(\Phi\setminus\{o\}\in\cdot\mid o\in\Phi).\nonumber
\end{equation}
We denote by $\Ex_{}^{{o}}$ the  expectation with respect to the Palm measure of $\Phi$
   and by $\Ex_{}^{!o}$ the  expectation with respect to the reduced Palm measure of $\Phi$ \cite[Chapter 8]{haenggi2012stochastic}.

\subsection{General Point Processes}

Let $c_d$ denote the volume of unit ball in $\mathbb{R}^d$.
The mean degree of the typical agent in a general point process is 
    \begin{align}
  \Ex^{o} \sum_{y\in\Psi}I(o,y)
       &\peq{a} \Ex \sum_{y\in\Psi}f(\|y\|)\nonumber
      \\
      &\peq{b} \mu c_d d \int_{0}^{\infty}f(r) r^{d-1} \dd r,\label{eq: EN_hub}
    \end{align}
    where step (a) follows from the independence of $\Phi$ and $\Psi$, and step (b) follows from Campbell's theorem. 
    By the mass transport principle \cite{baccelli:hal-02460214}, the mean degree of the typical hub (defined through the Palm expectation of $\Psi$) is $\lambda c_d d \int_{0}^{\infty}f(r) r^{d-1} \dd r.$ 
    For example, consider $d=2$, a general connection function $f$, and its dispersed version $f_{p}$. While the mean degree remains constant for any $p\in(0,1]$,
 the edges spreads out in space as $p\to 0$. To see this, consider the mean distance
 \begin{align}
 \Ex^{o} \sum_{y\in\Psi}I(o,y)\|y\| &=\int_{0}^{\infty}f_p(r)2\pi r^2 \dd r\nonumber \\
 & = \frac{1}{\sqrt{p}}\int_{0}^{\infty}f(r)2\pi r^2 \dd r,\nonumber
 \end{align}
which increases as the connection function becomes more dispersed.

 Let $M\triangleq\sum_{x\in\Phi^o\setminus\{o\}}\max_{y\in\Psi} I(o,y)I(x,y)$ denote the number of agents connected to $o$.  We have
    \begin{align}
        \Ex M &= \Ex\sum_{x\in\Phi^o\setminus\{o\} } 1-\prod_{y\in\Psi} \left(1-I(o,y)I(x,y)\right) \nonumber\\
        &= \Ex^{!o} \sum_{x\in\Phi} 1-\prod_{y\in\Psi} (1-f(\|y\|)f(\|x-y\|)).\nonumber
        \end{align}
    This follows from the fact that the edge indicator functions are independent Bernoulli random variables given $\Phi$ and $\Psi$.
    
    Let $N\triangleq\sum_{y\in\Psi}\sum_{x \in\Phi^o\setminus\{o\}} I(o,y)I(x,y)$ denote the total number of direct (two-edge) paths between $o$ and its connected agents. We have
    \begin{align}
       \Ex N  &= \Ex \sum_{y\in\Psi}\sum_{x \in\Phi^o\setminus\{o\}} I(o,y)I(x,y) \nonumber\\
       &= \Ex^{!o}\sum_{y\in\Psi}\sum_{x \in\Phi}f(\|y\|)f(\|x-y\|).\nonumber
    \end{align}

By definition, $N\geq M$, and the equality is achieved when $o$ and $x$ are connected through at most one hub, $\forall x\in\Phi$. The difference between $N$ and $M$ reflects how likely two agents connect through multiple hubs. The variance of $M$ is $\mathbb{V} M \triangleq \Ex M^2 - (\Ex M)^2$, which characterizes the heterogeneity of the number of connected agents.

\subsection{Poisson Point Processes}
\begin{figure*}[t]
\normalsize
\setcounter{MYtempeqncnt}{\value{equation}}
\setcounter{equation}{\value{equation}+2}
 \begin{align}
       \mathbb{V} N
       & = \Ex N+\lambda^2\mu\left(\int_{0}^{\infty}f(r)dc_d r^{d-1} \dd r\right)^3 \nonumber
       \\
       &\quad+\lambda\mu^2 \int_{\mathbb{R}^d}\int_{\mathbb{R}^d}\int_{\mathbb{R}^d}f(\|y_1\|)f(\|y_2\|) f(\|x-y_1\|)f(\|x-y_2\|) \dd y_1\dd y_2\dd x.\label{eq: E_N_o^2}
    \end{align}
    \setcounter{equation}{\value{MYtempeqncnt}}
\hrulefill
\vspace*{4pt}
\end{figure*}
Here we focus on the case where $\Psi$ and $\Phi$ are two independent PPPs in $\mathbb{R}^d$, defined on the same probability space. We derive the mean and variance of $N$ and $M$. For the PPP,
\[
{\rm P}_{\Phi}^{!o} = \rm P_{\Phi}.
\]
This is known as Slivnyak's Theorem  \cite{haenggi2012stochastic} and leads to  the result below.

    \begin{theorem}
    \label{thm: N,M}
For the PPP, the means of $M,~N$ are
\begin{equation}
  \Ex N = \lambda\mu \left(\int_{0}^{\infty}f(r)dc_d r^{d-1} \dd r\right)^2,        \label{eq: EN}
\end{equation}
and
 \begin{equation}
     \Ex M = \lambda\int_{\mathbb{R}^d}  1-\exp\left(-\mu\int_{\mathbb{R}^d}f(\|y\|)f(\|x-y\|)\dd y\right)\dd x.\label{eq: EM}
 \end{equation}
 Further, the variance of $N$ is given in (\ref{eq: E_N_o^2}),
 and the variance of $M$ satisfies $\mathbb{V} M \geq \Ex M$.
\end{theorem}
\begin{proof}
See Appendix \ref{appendix: N,M}.
   \end{proof}
   \begin{remark} 
       $\Ex M = {\Theta}(\mu)$\footnote{We use $f(x)=\Theta(g(x))$ as $x\to a$ to denote that $f(x)=O\left(g(x)\right)$ and $ g(x)=O(f(x))$ as $x\to a$.} as $\mu\to0$, since $1-\exp(-\mu) = {\Theta}(\mu)$ as $\mu\to0$.
   \end{remark}
 \begin{remark}
Eq. (\ref{eq: EN}) holds for general hub point process $\Psi$. The other equations in Theorem 1 rely on the fact that $\Psi$ is a PPP.  In (\ref{eq: EM}), $\int_{\mathbb{R}^d}f(\|y\|)f(\|x-y\|)\dd y$ only depends on $\|x\|$, as the rotation of $x$ is equivalent to rotating $y$, which has no effect on the integral. This observation is used in the next subsection.
 \end{remark} 
 \begin{remark}
  Both $N$ and $M$ are super-Poisson since $\Var N\geq \Ex N$, and $\Var M\geq \Ex M$. In (\ref{eq: E_N_o^2}), note that in the second term the integral to the power of 3 is larger than the triple integral in the third term. Hence, for fixed $\lambda\mu$, there exists a threshold of $\lambda$ above which increasing $\lambda/\mu$ increases the variance of $N$.
 \end{remark}

  \begin{figure*}[t]
       \centering
       \subfigure[Mean and the standard deviation of $N$. ]{           \psfrag{EN}{$\mathbb{E}N$}
       \psfrag{VN, lambda = 5, mu=50}{$\sqrt{\mathbb{V}N},\lambda=5,\mu=50$}
              \psfrag{VN, lambda = 50, mu=5}{$\sqrt{\mathbb{V}N},\lambda=50,\mu=5$}

\includegraphics[width=.42\textwidth]{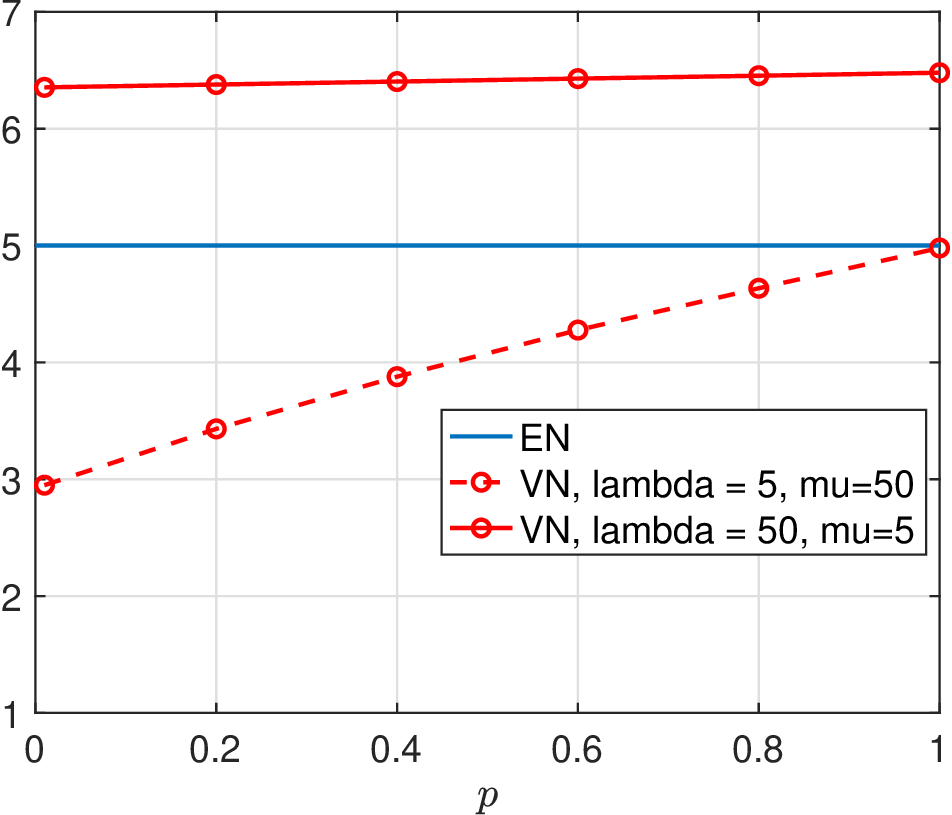}
         \label{fig: N_o}}
     \hfill
     \subfigure[Mean and the standard deviation of $M$. ]{ 
     \psfrag{EM, rho = 01}{$\mathbb{E}M,\lambda=5,\mu=50$}
          \psfrag{EM}{$\mathbb{E}M,\lambda=50,\mu=5$}
          \psfrag{sqrt{{V}M}, rho = 0111}{$\sqrt{\mathbb{V}M},\lambda=5,\mu=50$}
  \psfrag{sqrt{{V}M}, rho = 10102}{$\sqrt{\mathbb{V}M},\lambda=50,\mu=5$}
\includegraphics[width=.42\textwidth]{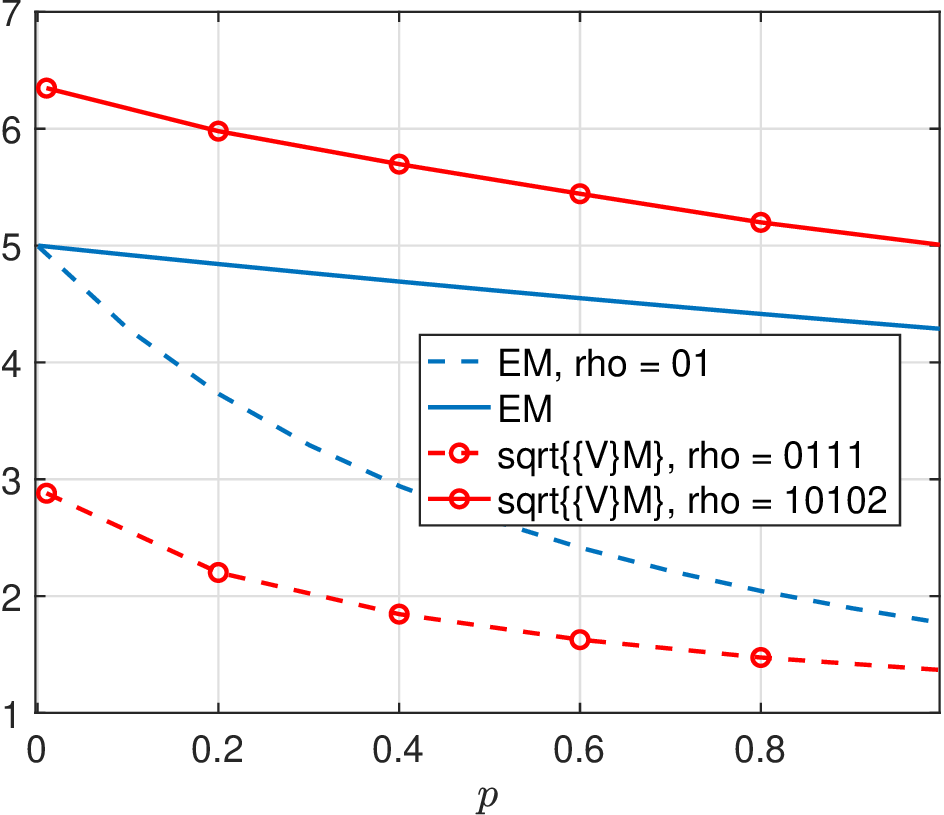}
         \label{fig: illu-g}}
           \caption{Mean and standard deviation of $N$ and $M$ for $f =\ind{(r\leq 0.2122)}$. We compare $(\lambda,\mu) = (5,50), (50,5)$ respectively. $\Ex N=5$. $\sqrt{\Var{M}}$ is obtained via simulation whereas the rest via Corollary 2.}
         \label{fig: N_oM_o_disk}
   \end{figure*}

  \begin{corollary}
  \label{cor: dispersed-mean-degree}
   For all $f$, $\Ex^{} M$ for $f_p$ monotonically decreases with $p$. For all $f$, $\Ex^{} M$ monotonically increases with $\lambda/\mu$ for fixed $\lambda\mu$. Further,
   
 \setcounter{equation}{\value{MYtempeqncnt}+3}
  \begin{equation}
 \lim_{p\to 0} \Ex^{} M = \Ex^{} N.
  \end{equation}
  \begin{equation}
 \lim_{\frac{\lambda}{\mu}\to \infty} \Ex^{} M = \Ex^{} N.
  \end{equation}
 \end{corollary}
 \begin{proof}
Let us write $\Ex^{} M$ for $f_p$  as
 \begin{equation}
\begin{split}\nonumber
   &\Ex^{} M\\
   & = \lambda\int_{\mathbb{R}^d}  1-\exp\left(-\mu p^2\int_{\mathbb{R}^d}f(\sqrt[d]{p}\|y\|)f(\sqrt[d]{p}\|x-y\|)\dd y\right)\dd x\\
   &  \peq{a} \frac{\lambda}{p}\int_{\mathbb{R}^d}   1-\exp\left(-\mu p\int_{\mathbb{R}^d}f(\|v\|)f(\|u-v\|)\dd v\right)\dd u.
 \end{split}
 \end{equation}
Step (a) follows from change of variables $u=\sqrt{p}x,~v=\sqrt{p}y$.
The last equation monotonically decreases with the increase of $p$ since $(1-\exp(-C x))/x$ monotonically decreases with $x$.
Further, this shows that for $\Ex M$, stretching $f$ to $f_p$ is equivalent to scaling the network by $\lambda' = \lambda/p$ and $\mu'=\mu p$.
As $p\to0$, all terms in the series expression of the integrand containing $p$ go to 0, hence the convergence to $\Ex N$.
 \end{proof}

\begin{remark}Recall that $N=M$ when each pair of agents are connected through at most one hub. 
 For fixed $\lambda\mu$, as $\lambda/\mu\to\infty$, the number of hubs for the typical agent decreases per (\ref{eq: EN_hub}). Intuitively, most agents only have an edge to their closest hub. 
 On the other hand, as $p\to0$, two agents in proximity are less unlikely to be connected through common hubs, and geometry has a diminishing impact on connections. The graph eventually becomes an arbitrarily large abstract random graph.  Since $N-M$ is a non-negative integer, Corollary 1 also implies the convergence of $M$ to $N$ in probability.
\end{remark}

 \subsection{Examples}
   In this subsection, we consider an example with $d=2$ and two special connection functions. 
    
    The first is\begin{equation}
  \nonumber  f_1(r) = \ind{(r\leq \theta)},\quad r\geq 0,\label{eq: disk}
      \end{equation} 
      where $\theta>0$ is the cut-off radius for drawing an edge between a agent and a hub. Two agents potentially have common hubs only if the intersection of their disks of radius $\theta$ are nonempty regardless of $\Psi$.
      
The second is
   \begin{equation}
\nonumber        f_2(r) =\frac{1}{2}\exp(-r/\theta),\quad r\geq 0.
\label{eq: exp dist}
\end{equation}
which has infinite support.

For both connection functions, the degree of the typical agent is Poisson distributed with mean 
    \begin{align}
      \Ex \sum_{y\in\Psi}I(o,y)&= \int_{0}^{\infty}\exp(- r/\theta)\pi\mu r \dd r\nonumber\\
      &=\int_{0}^{\infty} 2\pi\mu r \ind{(r\leq \theta)} \dd r\nonumber\\
      &= {\pi\mu}{\theta^2}.\nonumber\label{eq: EN_hub}
    \end{align}

\begin{corollary}
\label{cor: examples-N-M}
 For $f_1(r)= \ind{(r\leq \theta)}$, $\Ex N = \lambda\mu\pi^2\theta^4$, and 
 \begin{equation}
        \Ex M 
         = \int_{0}^{2\theta}2\pi\lambda v\left(1-e^{-\mu (2\theta^2\arccos(v/2\theta)-v\sqrt{\theta^2-v^2/4})}\right)\dd v.
        \label{eq: E M_0 disk}
\end{equation}
For $f_2(r)= \exp(-r/\theta)/2$,  $\Ex N = {\lambda\mu\pi^2}{\theta^4},$
and
\begin{align}
        \Ex M = \int_{0}^{\infty}2\pi\lambda v\left(1-e^{-\frac{\mu\pi v}{16}{({2\theta K_1(v/\theta)}+vK_0(v/\theta)})}\right)\dd v.\label{eq: E M_0 exp}
        \end{align}
where $K_n(\cdot)$ is the modified Bessel function of the second kind.

\end{corollary}    
\begin{proof}
For $f_1$, note that only agents within radius $2\theta$ can be connected to the typical agent $o$. Then the expression follows from the area of the intersection of two disks with radius $\theta$. For $f_2$, see Appendix \ref{appendix: examples-N-M}. 
\end{proof}
\begin{remark}
The integrand in $\Ex M$ has bounded support if and only if the connection function $f$ has bounded support, in which case its support is twice the support of $f$. 
\end{remark}

The numerical evaluations of (\ref{eq: E M_0 disk})  and (\ref{eq: E M_0 exp}) take a few millisecond using Matlab. $\Ex M$ increases linearly with $\lambda$ and sublinearly with $\mu$. As $\mu\to0$, $\Ex M$ is a linear function of $\mu$ asymptotically (see Remark 1). 
Fig. \ref{fig: N_oM_o_disk} plots the mean and variance of $N,~M$ for a fixed $\Ex N$ and $f_p(r)= pf(\sqrt{p}r)$, $p\in(0,1]$.

\section{Percolation}
  
    \subsection{Definitions}
We consider the percolation of the RBG graph where $\Phi$ and $\Psi$ are two independent PPPs with intensities $\lambda$ and $\mu$, respectively. We say that an infinite graph percolates if there exists an infinitely connected component in the graph a.s., or equivalently, the probability of the origin belonging to an infinite component is larger than 0.
 From the point of view of disease spreading, it means that a positive fraction of the network can eventually get infected from a single agent (ignoring the delay of spreading and recovery).  The RBG graph  model is ergodic and has a unique infinite component with probability 0 or 1 following the arguments in \cite[Theorems 2.1\&6.3]{MeesterRonald1996CP}. 
    

For a fixed $f$, let the percolation region be defined as
\begin{equation}
    \mathcal{D} \triangleq \{(\lambda,\mu): \mathcal{G}\left(\Phi,\Psi, f\right) ~\mathrm{percolates}\}.\nonumber
\end{equation}
   The boundary of $\mathcal{D}$ is  symmetrical, i.e., \[(\lambda,\mu)\in\mathcal{D} \quad \Leftrightarrow \quad (\mu,\lambda)\in\mathcal{D},\]
where $\Leftrightarrow$ denotes an implication in both directions.
    To see this, note that having an infinite component formed from a subset of $\Phi$ is equivalent to having  an infinite component formed from a subset of $\Psi$. Further, $\Phi$ and $\Psi$ are both PPPs.
    
For any fixed $\mu$, let $\lambda_c(\mu)$ be defined as \[
        \lambda_c(\mu) \triangleq \inf \{\lambda: (\lambda,\mu)\in \mathcal{D}\},\]
and $\mu_c$ be defined as
    \begin{equation}
        \mu_c \triangleq \inf \{\mu: \lambda_c(\mu)<\infty\}.\nonumber
    \end{equation} 
 Here our notations differ from \cite{penrose_2014} due to the physical interpretation of $\lambda$ and $\mu$ and that we focus on the critical density of hubs. The change of notation is not fundamental due to the symmetry of $\lambda$ and $\mu$ in our model.

 Consider a random unipartite geometric graph on a  PPP $\Phi$ with intensity $\lambda$, $\mathcal{G}(\Phi,f)$. 
 Denote by $\zeta_f$ the critical density for the percolation of the graph $\mathcal{G}(\Phi,f)$, i.e.,
 \[
\zeta_f \triangleq \inf \{\lambda\colon \mathcal{G}\left(\Phi,f\right) ~\mathrm{percolates}\}
 \]
 
 For $\int_{0}^{\infty}f(r)r^{d-1}\dd r<\infty$,  it is known that \cite[Chapter 6]{MeesterRonald1996CP}
  \[
  0<  \zeta_f<\infty.
    \] 
    Further, there is at most one unbounded component.



    \subsection{Some Bounds on the Percolation Thresholds}

        \begin{lemma}  
    The percolation region $\mathcal{D}$ for $\mathcal{G}(\Phi,\Psi,f)$ satisfies 
\begin{equation}
\min\{\lambda+\mu\colon (\lambda,\mu)\in \mathcal{D}\}\geq \zeta_f.
\end{equation}
    \end{lemma}
\begin{proof}
        Construct the following coupling: let $\mathcal{G}\left(\Phi\cup\Psi,f\right)$ be a unipartite random geometric graph with nodes  being the union of agents and hubs from $\mathcal{G}\left(\Phi,\Psi,f\right)$. The edge set of $\mathcal{G}\left(\Phi\cup\Psi,f\right)$ contains all edges of $\mathcal{G}\left(\Phi,\Psi,f\right)$, as well as those edges within agents and within hubs generated through $f$. So  $\mathcal{G}\left(\Phi\cup\Psi,f\right)$ dominates $\mathcal{G}\left(\Phi,\Psi,f\right)$ in the sense of the coupling just mentioned. So the percolation of $\mathcal{G}\left(\Phi,\Psi,f\right)$ implies the percolation of  $\mathcal{G}\left(\Phi\cup\Psi,f\right)$.
\end{proof}  
\begin{remark}
A special case for $f(r) = \ind (r\leq a) $ is proved in \cite{iyer2012percolation} with the same construction of coupling.    
\end{remark}

With a slight abuse of notation, let $\zeta(2a)$ denote the percolation threshold of $\mathcal{G}(\Phi, \ind (r\leq 2a))$ (the Poisson Boolean model). 

\begin{lemma}{\cite[Theorem 1.1]{penrose_2014}}\label{lemma: penrose}
    For the RBG model with $f(r) = \ind (r\leq a)$,      \[\mu_c = \zeta(2a).\]   
\end{lemma}
\begin{proof}
The proof is given in \cite{penrose_2014} through constructing a coupling of the RBG model with Bernoulli site percolation.
\end{proof}
\begin{remark}
    A special case for $d=2$ is proved in \cite{iyer2012percolation}.
\end{remark}
Here we prove that the same threshold holds for the connection functions $f(r) = p\ind (r\leq a)$, $p\in(0,1]$, which generalizes the Boolean connection function in \cite{iyer2012percolation,penrose_2014}. 

 \begin{theorem}
 For the RBG model with the connection functions $f(r) = p\ind (r\leq a)$, $p\in(0,1]$, 
 \label{thm: perco-th-p-boolean}
\begin{equation}
\mu_c = \zeta(2a).
 \end{equation}

 \end{theorem}
\begin{proof}
   The proof is based on the proof in \cite[Theorem 1.1]{penrose_2014}, with one modification. For the Bernoulli site percolation of $\epsilon\mathbb{Z}^2$ with retaining probability $q<1$ coupled with $\mathcal{Q}_{\mu}$ (corresponding to $\Phi$ with density $\lambda$ in our notation), we let $\lambda$ be such that 
   $1-\exp(-\lambda p^2 \epsilon^d) = q$ rather than requiring that $1-\exp(-\lambda \epsilon^d) = q$. The modification makes sure that if we independently remove points in $\Phi$ with probability $p^2$ (to account for the sampling of the two-edge path), the probability of having a point within each cube of size $\epsilon^d$ is  $q$.
\end{proof}

\begin{corollary}    \label{cor: bound_lambda_c_general_f}
    For a monotone decreasing $f$ with bounded support on $[0,a]$ and  lower bounded away from 0, 
    \[\mu_c = \zeta(2a).\]
\end{corollary}
\begin{proof}
Let $\underline{f}(r) = p\ind(r\leq a)$ where $p$ is the lower bound of $f$ on the support and $\bar{f}(r)=\ind(r\leq a)$. Then  there exists a coupling between the corresponding graphs;  $\mathcal{G}\left(\Phi,\Psi,\bar{f}\right)$ dominates $\mathcal{G}\left(\Phi,\Psi,f\right)$, and $\mathcal{G}\left(\Phi,\Psi,f\right)$ dominates $\mathcal{G}\left(\Phi,\Psi,\underline{f}\right)$. 
Since for both $\underline{f}$ and $\bar{f}$, $\mu_c  = \zeta(2a)$ by Theorem \ref{thm: perco-th-p-boolean} and Lemma \ref{lemma: penrose},  the proof is complete.
\end{proof}
\begin{remark}
\label{remark: bounded_travel} 
    Corollary \ref{cor: bound_lambda_c_general_f} shows that if  the traveling distance is bounded, it suffices to limit the density of hubs to prevent the disease from spreading, regardless of the density of agents.   
\end{remark}
Corollary \ref{cor: bound_lambda_c_general_f} also implies that the percolation threshold $\mu_c<\infty$ for RBG graphs with monotone decreasing connection functions with infinite support. This follows directly from the monotonicity of $\zeta(2a)$ in $a$ and the fact that $\lim_{a\to\infty}\zeta(2a)=0$.

\begin{theorem}
    For general monotone decreasing connection function $f$, 
    \begin{equation}    
    \lambda\mu\geq \left(\int_{\mathbb{R}^d}f(\|x\|) \dd x)\right)^{-2},\quad (\lambda,\mu)\in\mathcal{D}.
    \label{eq: g-w-branching-bound}
    \end{equation}
    \label{thm: g-w-branching-bound}
\end{theorem}
\begin{proof}
    We use the same approach as in \cite[Chapter 6]{MeesterRonald1996CP} that proves the existence of the threshold for Poisson random connection models. The idea is to establish a coupling of the RBG graph with a Galton-Watson branching process with average descendants \[\Ex N=\lambda\mu\left(\int_{\mathbb{R}^d}f(\|x\|) \dd x\right)^2.\]

    Let us start with the agent at the origin. In the first step, its connected hubs follow an inhomogeneous PPP with intensity function $\mu f(\|y\|)$, whose mean number is $\mu\int_{\mathbb{R}^d}f(\|y\|) \dd y$. In the second step, consider for each hub $y_1, y_2,...$ generated in the first step: for $y_1$, its connected agents in the second step  follow an inhomogeneous PPP  $\lambda f(\|x-y_1\|)$; for $y_2$, the connected agents in this step  follow an independent inhomogeneous PPP $\lambda f(\|x-y_2\|)(1-f(\|x-y_1\|))$, and so on. The agent set generated by each hub in this step is upper bounded by an independent inhomogeneous PPP of intensity $\lambda f(\|x-y_i\|),~i=1,2,...$. The total number of agents in the second step can be dominated by a compound Poisson r.v. whose mean is $\lambda\mu(\int_{\mathbb{R}^d}f(\|x\|) \dd x)^2 = \Ex N$. Repeating the two steps, a coupling between the RBG graph and a Galton-Watson branching process is established. Requiring $\Ex N \geq 1$ leads to Theorem \ref{thm: g-w-branching-bound}. 
\end{proof}

Theorem \ref{thm: g-w-branching-bound} uses a Galton-Watson branching process to lower bound the RBG graph, where the analogous approach applies to its unipartite counterpart. Considering the limiting result established for unipartite graphs in \cite{penrose1993spread}, one may conjecture that in the limit of dispersion, i.e., $p\to0$, the bound also becomes tight for RBG graphs. 

\begin{lemma}
For the RBG graph with general monotone decreasing connection functions,   $\mu_c>0$ if and only if $f$ has bounded support.
\label{lemma: iff_bounded_support}
\end{lemma}
\begin{proof}
If $f$ has bounded support, let $\bar{f}(r) = \ind(f(r)>0)$. By coupling, $\mathcal{G}(\Phi,\Psi,f)$ is dominated by $\mathcal{G}(\Phi,\Psi,\bar{f})$. We have $\mu_c \geq \bar{\mu}_c >0$.

   Now consider $f$ with infinite support. Consider a Poisson Boolean model. For all $\epsilon>0$, one can find $a_{\epsilon}$ large enough such that for $\lambda=\epsilon$ and $g(r)= \ind(r\leq 2a_{\epsilon})$, $\mathcal{G}(\Phi,g)$ percolates. By the infinite support, one can find $b\geq a_{\epsilon} $ such that $f(b)>0$.  Using Corollary \ref{cor: bound_lambda_c_general_f}, $\mu_c\leq \zeta(2a_{\epsilon})\leq \epsilon$.
\end{proof}
    Lemma \ref{lemma: iff_bounded_support} shows that if agents' traveling distance is unbounded, an arbitrary small density of hubs can cause disease spreading given that the density of agents is large enough. Combining this with Remark \ref{remark: bounded_travel}, we show that long-distance travels (e.g., via airplanes) play a critical role in disease spreading.



\subsection{Dispersed RBG Graph}
Let $f_p$ be the dispersed version of a connection function $f$ as defined earlier. For any fixed $\mu$, let $\lambda_c(\mu)$ denote the percolation threshold for the RBG graph with connection function $f$ and $\lambda_c^p(\mu)$ denote that for the RBG graph with $f_p$. $\mu_c$ and $\mu_c^p$ are defined similarly. For the random geometric unipartite graph, it is shown in \cite{Franceschetti05continuumpercolation} that $\zeta_f\geq \zeta_{f_p}$. Later strict inequality is proved in some cases. 
 Here we show that a similar result holds for the RBG graph.

     

 \begin{theorem}
 For all $f$, $p\in(0,1]$, 
\begin{equation}
\lambda_c(\mu)\geq \lambda_c^p(\mu),\nonumber
\end{equation} 
and consequently,
\[\mu_c^p\leq\mu_c.\]
 \end{theorem}
\begin{proof}
Firstly, consider any fixed graph. Consider its site percolation and bond percolation respectively \cite[Chapter 1]{grimmett1999percolation} as follows: the former  independently removes points in $\Phi$ and $\Psi$ with probability $1-p$; the latter independently removes edges with probability $1-p$. The bond percolation threshold is no greater than the site percolation \cite{Franceschetti05continuumpercolation}.
Now, consider any realization of $\mathcal{G}$ generated from $\Phi'=\Phi/p,~\Psi'=\Psi/p,~f(r)$. $\lambda_c(\mu)$ is the site percolation threshold (with probability $p$) of the graph $\mathcal{G}(\Phi',\Psi',f)$, and  equivalently, the percolation threshold for $\mathcal{G}(\Phi,\Psi,f)$. $\lambda_c^p(\mu)$ is the bond percolation threshold of the graph $\mathcal{G}(\Phi',\Psi',pf(r))$, which is equivalent to the percolation of $\mathcal{G}(\Phi,\Psi,~pf(\sqrt{p}r))$ by  scaling. The second inequality follows from the definition of $\mu_c$.
\end{proof}
  A dispersive function lowers the percolation threshold and expands the percolation region. This is consistent with the monotonicity and the limit of the mean degree in Corollary \ref{cor: dispersed-mean-degree}.
One may expect (\ref{eq: g-w-branching-bound}) to become more tight when $p\to0$ ($f$ is dispersive).  In the case where (\ref{eq: g-w-branching-bound}) is tight, we have \begin{align}
        \zeta_f &\leq \min_{(\lambda,\mu)\in\mathcal{D}}\lambda+\mu\nonumber \\
        &\approx \min_{(\lambda,\mu)\in\mathcal{D}} \lambda + {\lambda^{-1}\left(\int_{\mathbb{R}^d}f(\|x\|) \dd x\right)^{-2}} \nonumber\\
        &= {2}{\left(\int_{\mathbb{R}^d}f(\|x\|) \dd x\right)}^{-1}.\nonumber
\end{align}


    
\section{Discussion}
\subsection{Dependent RBG Models}
For epidemic modeling, RBG models also admit cases where there is dependence between $\Phi$ and $\Psi$, i.e., correlated hub and agent locations. From the degree point of view, one would expect a higher mean degree $\mathbb{E}M$ than in the independent case when there is clustering/attraction between $\Phi$ and $\Psi$ and lower when there is repulsion. From the percolation point of view, the impact of such dependence is less straightforward.  
 Consider the Poisson cluster point process, with $\Psi$ the parent process and $\Phi$ the daughter process. One expects that ``too much clustering'' would prevent spreading from one cluster to another. For the repulsive type of dependence, for instance, agents could be “repelled” from hubs by placing them at the Voronoi corners of the hubs' Voronoi tessellation.  This may not be practical but could be interesting as a theoretical thought experiment - now the distances from hub to agents are relatively larger, but each agent has three hubs at the same distance, which could spur growth. Considering both situations, we expect that the attractive case makes local spreading easier but global spreading (percolation) harder, while the repulsive case works the other way around. Of course these are very qualitative statements; for a rigorous analysis, one would have to properly rescale or normalize the densities and the connection function. 
\subsection{Dynamics}
The static RBG model in this work can be interpreted as capturing discrete-time disease dynamics as follows: in each
time step, the infectious agents infect those connected to the same hub(s) and
then recover in the next time step(s). Since all edges are i.i.d., the network percolates
effectively if starting with an infected agent at the origin from time 0, eventually a
non-trivial fraction of the entire population is infected at least once.  

An important direction for future work is to incorporate general temporal dynamics into this static model. For example, one may consider dynamics in edges and agent locations as well as in agents' states of infection and recovery. In this direction, one can consider the spatiotemporal graph models as in \cite[Chapter 11]{haenggi2012stochastic} to trace the paths and delays of disease spread. Another way is to introduce susceptible-infected-susceptible dynamics of agents' states based on the RBG graph and study its steady state behaviors \cite[Chapter 7]{liggett1985interacting}.

\section{Conclusion Remarks}
This work proposes a RBG graph model for disease spreading via hubs. We study the joint effect of the agent density, hub density, and connection function. The existence of a critical hub density depends only on the boundedness of the support of the connection function, which relates to curbing the traveling distance of individuals. When it comes to dispersion, both the degree distribution and the percolation threshold suggest that increasing dispersion helps spread the disease. The percolation properties of RBG graphs relate to unipartite graphs with modified connection functions. 
An interesting question in this direction is if and when the properties of the RBG graphs can be well represented by unipartite graphs with some modified connection functions. Our conjecture is that for independent connections between different pairs of agents, such representation is unlikely due to the oblivion of the local dependence (present in the RBG models). 
 Another direction is to consider hybrid models where agents may get infected either through common hubs or direct interactions between agents. The former infection mechanism is more centralized than the latter.

\appendix
\subsection{Proof of Theorem \ref{thm: N,M}}
\label{appendix: N,M}
 \begin{align}
       \Ex N  
       &=\Ex\sum_{y\in\Psi}\sum_{x \in\Phi} f(\|y\|)f(\|x-y\|)\nonumber\\
       &\peq{a}\lambda\mu \int_{\mathbb{R}^d}\int_{\mathbb{R}^d}  f(\|y\|)f(\|x-y\|) \dd x \dd y\nonumber\\
       & = \lambda\mu \left(\int_{0}^{\infty}f(r)dc_d r^{d-1} \dd r\right)^2.\nonumber
    \end{align}
 Step (a) follows from Campbell's theorem and the independence of $\Phi$ and $\Psi$. 
 \begin{align}
            \nonumber\Ex M & =  
            \Ex\sum_{x \in\Phi} \max_{y\in\Psi}I(o,y)I(x,y)\nonumber \\
            & = \Ex\sum_{x\in\Phi}1-\prod_{y\in\Psi} (1-f(\|y\|)f(\|x-y\|))\\
            & \peq{a}\lambda \int_{\mathbb{R}^d}1- \Ex\prod_{y\in\Psi} \left(1-f(\|y\|)f(\|x-y\|)\right) \dd x \nonumber\\
     & \peq{b}  \lambda\int_{\mathbb{R}^d}  1-\exp\left(-\mu\int_{\mathbb{R}^d}f(\|y\|)f(\|x-y\|)\dd y\right) \dd x.\nonumber
   \end{align}   
   Step (a) again follows from Campbell's theorem. Step (b) follows from the probability generating functional (PGFL) of the PPP.
   \begin{align}
  & \Ex N^2 \nonumber \\
   &= \Ex \left(\sum_{y\in\Psi}\sum_{x \in\Phi} I(o,y)I(x,y)\right)^2\nonumber\\
       & = \Ex \sum_{y\in\Psi}\sum_{x \in\Phi} I(o,y)I(x,y) \nonumber\\
       &\quad+\Ex\sum_{y\in\Psi}\sum_{x_1,x_2\in\Phi}^{\neq} I(x_1,y)I(o,y)I(x_2,y)\nonumber\\
       &\quad+\Ex\sum_{x\in\Phi}\sum_{y_1,y_2\in\Psi}^{\neq} I(x,y_1)I(o,y_1)I(x,y_2)I(o,y_2)\nonumber\\
       &\quad+\Ex\sum_{y_1,y_2\in\Psi}^{\neq}\sum_{x_1,x_2\in\Phi}^{\neq} I(x_1,y_1)I(o,y_1)I(x_2,y_2)I(o,y_2).\nonumber
   \end{align}
Here $\sum_{x_1,x_2\in\Phi}^{\neq}$ means that the summation takes over all ordered pairs of distinct points of $\Phi$. Observe that the first term is $\Ex N$ and the last term is equivalent to $(\Ex N)^2$. Using the joint density of $\Phi$ and $\Psi$ as well as the reduced second moment density of $\Phi$ and $\Psi$ we obtain (\ref{eq: E_N_o^2}).   

   For $M$,
   \begin{align}
        &\Ex M^2\nonumber \\
        &= \Ex \left(\sum_{x \in\Phi} \max_{y\in\Psi}I(o,y)I(x,y) \right)^2\nonumber\\
            & = \Ex \sum_{x\in\Phi} \max_{y\in\Psi}I(o,y)I(x,y) \nonumber \\
            &\quad+ \Ex\sum_{x_1,x_2\in\Phi}^{\neq}\max_{y_1\in\Psi}I(o,y_1)I(x_1,y_1)\max_{y_2\in\Psi}I(o,y_2)I(x_2,y_2)\label{eq: last}\\
            & \geq \Ex M + (\Ex M)^2.\nonumber
   \end{align}
   The last step follows from the following facts: the first term of (\ref{eq: last}) $\Ex \sum_{x\in\Phi} \max_{y\in\Psi}I(o,y)I(x,y) = \Ex M$; for the second term of  (\ref{eq: last}), it is equal to $(\Ex M)^2$ if  
$\max_{y_1\in\Psi}I(o,y_1)I(x_1,y_1)$ and $\max_{y_2\in\Psi}I(o,y_2)I(x_2,y_2)$ are independent. When the max over $\Psi$ is possible with $y_1=y_2$, there exists a positive correlation between them, hence the inequality.

\subsection{Proof of Corollary \ref{cor: examples-N-M}}
\label{appendix: examples-N-M}
The fact that $\Ex N = \lambda\mu {\pi^2}{\theta^4}$ directly follows from Theorem 1 and the integration of $f_1$.  To calculate $\Ex M$, we first present a useful lemma.      \begin{lemma}
For a stationary PPP $\Psi\subset \mathbb{R}^2$ with density $\mu$ and a deterministic $x\in\mathbb{R}^2$, the point process $\{y\in\Psi: \|y\|+\|x-y\|\}$ is a 1D PPP with intensity function
\begin{equation}
\lambda_x(r) = \frac{\mu\pi}{4}\frac{ 2r^2-\|x\|^2}{\sqrt{r^2-\|x\|^2}}, \quad r\geq \|x\|.    
\label{eq: lambda_x}
\end{equation}
   \end{lemma}
   \begin{proof}
   First note that $\|y\|+\|x-y\|\geq\|x\|$ by the triangular inequality. So the points lie in $[\|x\|,\infty)$. $\|y\|+\|x-y\|\leq r$ is a closed region whose boundary is an ellipse. Its area only depends on $\|x\|$ and $r$. The intensity measure is
  \[\Lambda([\|x\|,r])= \Ex \sum_{y\in\Psi} \ind(\|y\|+\|x-y\|\leq r) = \frac{\mu\pi r}{4}\sqrt{r^2-\|x\|^2}.\] 
  Taking its derivative, we obtain \eqref{eq: lambda_x}. The Poisson property is preserved as the numbers of points in disjoint regions are independent.
 \end{proof}

Now,
\begin{align}
  &\Ex \max_{y\in\Psi} I(o,y)I(x,y)\nonumber \\
     &=\Ex_{\Psi}\left[ 1- \prod_{y\in\Psi} \left(1-\frac{1}{4}\exp\big(- (\|y\|+\|x-y\|)/\theta\big)\right)\right]\nonumber\\
     &\peq{a} 1-\exp\left(-\int_{\|x\|}^{\infty} \lambda_{x}(r) \frac{1}{4}\exp(-r/\theta)\dd r\right)\nonumber\\
     &=1-\exp\left(-\frac{\mu\pi \|x\|^2}{16}\left(\frac{2\theta K_1(\|x\|/\theta )}{ \|x\|}+K_0\big( \|x\|/\theta \big)\right)\right)\nonumber.
   \end{align}
    Step (a) follows from Lemma 1 and the PGFL of the PPP.

\printbibliography
\end{document}